\documentclass[a4paper,UKenglish,cleveref, autoref, thm-restate,11pt]{article}
\usepackage[utf8]{inputenc}
\usepackage{fullpage}
\usepackage{xcolor}
\usepackage{amsmath,amsfonts,amsthm,amssymb}
\usepackage{hyperref}
\usepackage[capitalise]{cleveref}
\usepackage{verbatim}
\usepackage{enumerate}
\usepackage{thm-restate}
\usepackage{algorithm2e}

\usepackage{authblk}

\usepackage{tikz}
\usetikzlibrary{decorations.pathreplacing}
\usetikzlibrary{snakes}
\definecolor{darkblue}{rgb}{0.2,0.2,0.6}

\def\dd{\mathinner{.\,.}}
\newcommand{\cO}{\mathcal{O}}
\newcommand{\Oh}{\cO}

\newcommand{\per}{\textsf{per}}

\newcommand{\iden}{\textsf{id}}

\newcommand{\F}{\mathcal{F}}

\newcommand{\ttop}{\mathit{top}}
\newcommand{\ppop}{\mathit{pop}}
\newcommand{\ppush}{\mathit{push}}

\newcommand{\Whites}{\mathit{Whites}}

\newcommand{\hper}{\textsf{hper}}
\newcommand{\vper}{\textsf{vper}}

\newcommand{\width}{\textsf{width}}
\newcommand{\height}{\textsf{height}}

\newcommand{\defproblem}[3]{
  \vspace{2mm}
\noindent\fbox{
  \begin{minipage}{0.96\textwidth}
  \textsc{#1}\\
  {\bf{Input:}} #2  \\
  {\bf{Output:}} #3
  \end{minipage}
  }
  \vspace{2mm}
}

  \theoremstyle{plain}
  \newtheorem{theorem}{Theorem}
  \newtheorem{lemma}{Lemma}  
  \newtheorem{corollary}[theorem]{Corollary} 
    
  \newtheorem{fact}{Fact}
  \newtheorem{observation}{Observation}
  \theoremstyle{definition}
  \newtheorem{definition}{Definition}

  \newtheorem{remark}[definition]{Remark}
  \newtheorem{claim}{Claim}

\usepackage{mathtools}
\DeclarePairedDelimiter{\floor}{\lfloor}{\rfloor}

\newcommand{\SP}{\mathit{SpecialPoints}}
\newcommand{\Thick}{\mathit{ThickQuartics}}

\title{The number of repetitions in 2D-strings}

\author[1,2]{Panagiotis Charalampopoulos}
\author[2]{Jakub Radoszewski}
\author[2]{Wojciech Rytter}
\author[2]{Tomasz Waleń}
\author[2]{Wiktor Zuba}

\affil[1]{Department of Informatics, King's College London, UK\\
    \texttt{panagiotis.charalampopoulos@kcl.ac.uk}}
\affil[2]{Institute of Informatics, University of Warsaw, Poland\\
    \texttt{$\{$jrad,rytter,walen,w.zuba$\}$@mimuw.edu.pl}}

\date{\vspace{-5ex}}

\begin{document}

\maketitle

\thispagestyle{empty}

\begin{abstract}
The notions of periodicity and repetitions in strings, and hence these of runs and squares, naturally extend to two-dimensional strings.
We consider two types of repetitions in 2D-strings: \emph{2D-runs} and \emph{quartics} (quartics are a 2D-version of squares in standard strings).
Amir et al.~introduced 2D-runs, showed that there are $\cO(n^3)$ of them in an $n \times n$ 2D-string and presented a simple construction
giving a lower bound of $\Omega(n^2)$ for their number (\emph{Theoretical Computer Science}, 2020).
We make a significant step towards closing the gap between these bounds by showing that the number of 2D-runs in an $n \times n$ 2D-string is $\cO(n^2 \log^2 n)$.
In particular, our bound implies that the $\cO(n^2\log n + \textsf{output})$ run-time of the algorithm of Amir et al.\ for computing 2D-runs is also $\cO(n^2 \log^2 n)$.
We expect this result to allow for exploiting 2D-runs algorithmically in the
area of 2D pattern matching.

A quartic is a 2D-string composed of $2 \times 2$ identical blocks (2D-strings) that was introduced by Apostolico and Brimkov (\emph{Theoretical Computer Science}, 2000), where by
quartics they meant only \emph{primitively rooted} quartics, i.e.~built of a primitive block.
Here our notion of quartics is more general and analogous to that of squares in 1D-strings.
Apostolico and Brimkov showed that there are $\cO(n^2 \log^2 n)$ occurrences of primitively rooted quartics in an $n \times n$ 2D-string and that this bound is attainable. Consequently the number of 
distinct primitively rooted quartics is $\cO(n^2 \log^2 n)$. 
The straightforward bound for the maximal number of distinct general quartics is $\cO(n^4)$.
Here, we prove that the number of distinct general quartics is also $\cO(n^2 \log^2 n)$. This extends the rich combinatorial study of the number of distinct squares in a 1D-string, that was initiated by Fraenkel and Simpson (\emph{Journal of Combinatorial Theory, Series A}, 1998), to two dimensions.

Finally, we show some algorithmic applications of 2D-runs.
Specifically, we present algorithms for computing all occurrences of primitively rooted quartics and 
counting all general distinct quartics in $\cO(n^2 \log^2 n)$ time, which is quasi-linear with respect to
the size of the input. The former algorithm is optimal due to the lower bound of Apostolico and Brimkov. The latter can be seen as a continuation of works on enumeration of distinct squares in 1D-strings using runs (Crochemore et al., \emph{Theoretical Computer Science}, 2014). However, the methods used in 2D are different because of different properties of 2D-runs and quartics. 
\end{abstract}

\clearpage
\setcounter{page}{1}

\section{Introduction}
Periodicity is one of the main and most elegant notions in stringology.
It has been studied extensively both from the combinatorial and the algorithmic perspective; see e.g.~the books~\cite{DBLP:books/daglib/0020103,DBLP:books/daglib/0020111,DBLP:books/daglib/0019130}.
A classic combinatorial result is the periodicity lemma due to Fine and Wilf~\cite{fine1965uniqueness}.
From the algorithmic side, periodicity often poses challenges in pattern matching, due to the following fact: a pattern $P$ can have many occurrences in a text $T$ that are ``close'' to each other if and only if $P$ has a ``small'' period.
On the other hand, the periodic structure indeed allows us to overcome such challenges; see~\cite{DBLP:books/daglib/0020103,DBLP:books/daglib/0020111}.

Runs, also known as maximal repetitions, are a fundamental notion in stringology.
A run is a periodic fragment of the text that cannot be extended without changing the period. Runs were introduced in~\cite{DBLP:journals/tcs/IliopoulosMS97}. Kolpakov and Kucherov presented an algorithm to compute all runs in a string in time linear with respect to the length of the string over a linearly-sortable alphabet~\cite{KK:99}.
Runs fully capture the periodicity of the underlying string and, since the publication of the algorithm for their linear-time computation, they have assumed a central role in algorithm design for strings.
They have been exploited for text indexing~\cite{DBLP:conf/stoc/KempaK19}, answering internal pattern matching queries in texts~\cite{DBLP:conf/isaac/Charalampopoulos19,DBLP:conf/soda/KociumakaRRW15}, or reporting repetitions in a string~\cite{DBLP:conf/esa/AmirBCK19,DBLP:conf/cpm/Charalampopoulos20,DBLP:journals/tcs/CrochemoreIKRRW14}, to name a few applications.

Kolpakov and Kucherov also posed the so-called runs conjecture which states that there are at most $n$ runs in a string of length $n$.
A long line of work on the upper~\cite{DBLP:conf/mfcs/CrochemoreI07,DBLP:journals/jcss/CrochemoreI08,DBLP:journals/tcs/CrochemoreIT11,DBLP:conf/lata/Giraud08,DBLP:journals/tcs/PuglisiSS08,DBLP:conf/stacs/Rytter06,DBLP:journals/iandc/Rytter07} and lower bounds~\cite{DBLP:journals/ijfcs/FranekY08,DBLP:conf/stringology/MatsubaraKIBS08,DBLP:journals/ajc/Simpson10}
was concluded by Bannai et al.\ who positively resolved the runs conjecture in~\cite{runstheorem} (see also an alternative proof in~\cite{DBLP:journals/tcs/CrochemoreM16} and a tighter upper bound for binary strings from \cite{DBLP:conf/spire/0001HIL15}).

A square is a concatenation of two copies of the same string. Fraenkel and Simpson~\cite{DBLP:journals/jct/FraenkelS98} showed that a string of length $n$ contains at most $2n$ distinct square factors. This bound was improved in~\cite{DBLP:journals/dam/DezaFT15,DBLP:journals/tcs/Ilie07}. 
All distinct squares in a string of length $n$ can be computed in $\cO(n)$ time assuming an integer alphabet~\cite{DBLP:conf/cpm/BannaiIK17,DBLP:journals/tcs/CrochemoreIKRRW14,DBLP:journals/jcss/GusfieldS04} (see \cite{DBLP:journals/tcs/StoyeG02} for an earlier $\cO(n \log n)$ algorithm).

Pattern matching and combinatorics on 2D strings have been studied for more than 40 years, see e.g.~\cite{DBLP:journals/siamcomp/AmirBF94,DBLP:conf/soda/AmirF91,DBLP:journals/siamcomp/Baker78a,DBLP:journals/ipl/Bird77,DBLP:books/daglib/0020103,DBLP:books/daglib/0020111}.
In this paper we consider 2-dimensional versions of runs,  introduced by Amir et al.~\cite{DBLP:conf/esa/AmirLMS18,Amir2020}, and of repetitions in 2D-strings, introduced by Apostolico and Brimkov~\cite{DBLP:journals/tcs/ApostolicoB00}. As discussed in~\cite{Amir2020,DBLP:journals/dam/ApostolicoB05}, one could potentially exploit such repetitions in a 2D-string, which could for instance be an image, in order to compress it.

A \emph{2D-run} in a 2D-string $A$ is a subarray of $A$ that is both horizontally periodic and vertically periodic and that cannot be extended by a row or column without changing the horizontal or vertical periodicity (a formal definition follows in Section~\ref{sec:prelim}); see \cref{fig:example}(a). Amir et al.~\cite{DBLP:conf/esa/AmirLMS18,Amir2020} have shown that the maximum number of 2D-runs in an $n \times n$ array is $\cO(n^3)$ and presented an example with $\Theta(n^2)$ 2D-runs. In \cite{Amir2020} they presented an $\cO(n^2 \log n + \textsf{output})$-time algorithm for computing 2D-runs.

A \emph{quartic} is a configuration that is composed of $2 \times 2$ occurrences of an array $W$ (see \cref{fig:example}(b)) and a \emph{tandem} is a configuration consisting of two occurrences of an array $W$ that share 
one side (Apostolico and Brimkov~\cite{DBLP:journals/tcs/ApostolicoB00} also considered another type of tandems, which share one corner; see also \cite{DBLP:conf/cpm/AmirBLMS20}). An array $W$ is called \emph{primitive} if it cannot be partitioned into non-overlapping replicas of some array $W'$.
Apostolico and Brimkov~\cite{DBLP:journals/tcs/ApostolicoB00} considered only quartics and tandems with primitive $W$ (we call them \emph{primitively rooted}) and showed tight asymptotic bounds $\Theta(n^2 \log^2 n)$ and $\Theta(n^3 \log n)$ for the maximum number of occurrences of such quartics and tandems in an $n \times n$ array, respectively. In \cite{DBLP:journals/dam/ApostolicoB05} they presented an optimal $\cO(n^3 \log n)$-time algorithm for computing all {\it occurrences} of tandems with primitive $W$. This extends a result that a 1D-string of length $n$ contains $\cO(n \log n)$ occurrences of primitively rooted squares and they can all be computed in $\cO(n \log n)$ time; see~\cite{DBLP:journals/ipl/Crochemore81,DBLP:journals/tcs/StoyeG02}.
In this paper we consider the  numbers of {\it all distinct}  quartics, which is a more complicated problem.

\begin{figure}[htpb]
    \centering
\begin{tikzpicture}[scale=0.45]
\begin{scope}
\foreach \x/\y/\c in {
0/-1/a,1/-1/a,2/-1/a,3/-1/a,4/-1/a,5/-1/a,6/-1/a,7/-1/a,
0/0/b,1/0/a,2/0/b,3/0/a,4/0/b,5/0/a,6/0/b,7/0/a,
0/1/a,1/1/a,2/1/a,3/1/a,4/1/a,5/1/a,6/1/a,7/1/a,
0/2/a,1/2/a,2/2/b,3/2/a,4/2/b,5/2/a,6/2/b,7/2/a,
0/3/a,1/3/b,2/3/a,3/3/a,4/3/a,5/3/b,6/3/a,7/3/b
}{
    \draw(\x,\y) node[above] {$\c$};
}
\draw[xshift=-0.5cm,yshift=-0.1cm,] (-0.2,-1.2) rectangle (8.2,4.2);
\draw[xshift=-0.5cm,brown,line width=0.4mm] (1,-1) rectangle (8,3);
\draw (3.5,-2) node {\textbf{(a)} a 2D-run};
\end{scope}
\begin{scope}[xshift=12cm]
\foreach \x/\y/\c in {
0/-1/a,1/-1/a,2/-1/a,3/-1/a,4/-1/a,5/-1/a,6/-1/a,7/-1/a,
0/0/b,1/0/a,2/0/b,3/0/a,4/0/b,5/0/a,6/0/b,7/0/a,
0/1/a,1/1/a,2/1/a,3/1/a,4/1/a,5/1/a,6/1/a,7/1/a,
0/2/a,1/2/a,2/2/b,3/2/a,4/2/b,5/2/a,6/2/b,7/2/a,
0/3/a,1/3/b,2/3/a,3/3/a,4/3/a,5/3/b,6/3/a,7/3/b
}{
    \draw(\x,\y) node[above] {$\c$};
}
\draw[xshift=-0.5cm,yshift=-0.1cm,] (-0.2,-1.2) rectangle (8.2,4.2);
\draw[xshift=-0.5cm,blue,line width=0.3mm] (1,-1) rectangle (3,1);
\draw[xshift=-0.5cm,blue,line width=0.3mm] (1,1) rectangle (3,3);
\draw[xshift=-0.5cm,blue,line width=0.3mm] (3,-1) rectangle (5,1);
\draw[xshift=-0.5cm,blue,line width=0.3mm] (3,1) rectangle (5,3);
\draw (3.5,-2) node {\textbf{(b)} a quartic};

\end{scope}
\end{tikzpicture}
    \caption{Examples of a 2D-run and a quartic.}
    \label{fig:example}
\end{figure}
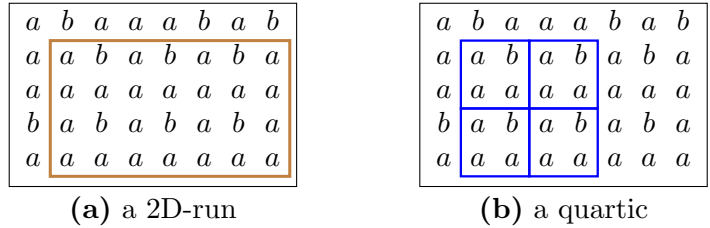

When computing 2D-runs we consider positioned runs: two 2D-runs with same content but
starting in different points are considered distinct. However in case of quartics,
similarly as in case of 1D-squares, we consider unpositioned
quartics; if two quartics have the same content but start in different positions, we consider them 
equal. 

\paragraph{Our Results:}
\begin{itemize}
    \item We show that the number of 2D-runs in an $n \times n$ array is $\cO(n^2 \log^2 n)$. This improves upon the $\cO(n^3)$ upper bound of Amir et al.~\cite{DBLP:conf/esa/AmirLMS18,Amir2020} and proves that their algorithm computes all 2D-runs in an $n \times n$ 2D-string in $\cO(n^2 \log^2 n)$ time {\bf (Section~\ref{sec:runs})}.
    \item We show that the number of distinct quartics in an $n \times n$ array is $\cO(n^2 \log^2 n)$.
    This can be viewed as an extension of the bounds on the maximum number of distinct square factors in a 1D-string~\cite{DBLP:journals/dam/DezaFT15,DBLP:journals/jct/FraenkelS98} {\bf (Section~\ref{sec:quartics})}.
    \item We present algorithmic implications of the new upper bound for 2D-runs. We show that all occurrences of primitively rooted quartics can be computed in quasi-linear, $\cO(n^2 \log^2 n)$ time, which is optimal by the bound
    of Apostolico and Brimkov~\cite{DBLP:journals/tcs/ApostolicoB00}.
    Thus our algorithm complements the result of Apostolico and Brimkov~\cite{DBLP:journals/dam/ApostolicoB05} who gave an optimal algorithm for computing all occurrences of primitively rooted tandems.
    We also show that all distinct quartics can be computed in quasi-linear, $\cO(n^2 \log^2 n)$ time, which extends efficient computation of distinct squares in 1D-strings~\cite{DBLP:conf/cpm/BannaiIK17,DBLP:journals/tcs/CrochemoreIKRRW14,DBLP:journals/jcss/GusfieldS04} to 2D {\bf (Section~\ref{sec:apps})}.
\item 
As an easy side result, we show tight $\Theta(n^3)$ bounds for the maximum number of distinct tandems in an $n \times n$ array 
and how to report  them in $\cO(n^3)$ time {\bf (Section~\ref{sec:prelim})}.
\end {itemize}

\section{Preliminaries}\label{sec:prelim}

\paragraph{1D-Strings.}
We denote by $[a,b]$ the set $\{i \in \mathbb{Z}: a \leq i\leq b\}$.
Let $S=S[1]S[2]\cdots S[|S|]$ be a \textit{string} of length $|S|$ over an alphabet $\Sigma$. The elements of $\Sigma$ are called \textit{letters}.
For two positions $i$ and $j$ on $S$, we denote by $S[i\dd j]=S[i]\cdots S[j]$ the \textit{fragment} of $S$ that starts at position $i$ and ends at position $j$ (it equals $\varepsilon$ if $j<i$).
A positive integer $p$ is called a \emph{period} of $S$ if $S[i] = S[i + p]$ for all $i = 1, \ldots, |S| - p$.  We refer to the smallest period as \emph{the period} of the string, and denote it by $\per(S)$.

\begin{lemma}[Periodicity Lemma (weak version), Fine and Wilf~\cite{fine1965uniqueness}]\label{lem:FW}
If $p$ and $q$ are periods of a string $S$ and satisfy $p + q \leq |S|$, then $\gcd(p, q)$ is also a period of $S$.
\end{lemma}

A string $S$ is called \emph{periodic} if $\per(S)\leq |S|/2$.
By $ST$ and $S^k$ we denote the concatenation of strings $S$ and $T$ and $k$ copies of the string $S$, respectively.
A string $S$ is called \emph{primitive} if it cannot be expressed as $U^k$ for a string $U$ and an integer $k>1$.

A string of the form $U^2$ for string $U$ is called a \emph{square}. A square $U^2$ is called \emph{primitively rooted} if $U$ is primitive. We will make use of the following important property of squares.

\begin{lemma}[Three Squares Lemma, \cite{DBLP:journals/algorithmica/CrochemoreR95}]\label{lem:3sq}
Let $U$, $V$ and $W$ be three strings such that $U^2$ is a proper prefix of $V^2$, $V^2$ is a proper prefix of $W^2$ and $U$ is primitive. Then $|U|+|V|\leq|W|$.
\end{lemma}

A \emph{run} (also known as \emph{maximal repetition}) in $S$ is a periodic fragment $R=S[i\dd j]$ which cannot be extended either to the left or to the right without increasing the period $p= \per(R)$, i.e.~if $i>1$ then $S[i-1]\neq S[i +p-1]$ and if $j<|S|$ then $S[j+1] \neq S[j-p+1]$.
Let $\mathcal{R}(S)$ denote the set of all runs of string $S$.
For periodic fragment $U=S[a\dd b]$, the run that extends $U$ is the unique run $R=S[i \dd j]$ such that $i\leq a \leq b \leq j$ and $\per(R)=\per(U)$. An occurrence of a square $U^2$ is said to be \emph{induced} by a run $R$ if $R$ extends $U^2$. Every square is induced by exactly one run~\cite{DBLP:journals/tcs/CrochemoreIKRRW14}.

\paragraph{2D-Strings.}

Let $A$ be an $m \times n$ array (2D-string).
We denote the height and width of $A$ by $\height(A)=m$ and $\width(A)=n$, respectively.
By $A[i,j]$ we denote the cell in the $i$th row and $j$th column of $A$; see \cref{fig:power}(a). By $A[i_1 \dd i_2,j_1 \dd j_2]$ we denote the subarray formed of rows $i_1,\ldots,i_2$ and columns $j_1,\ldots,j_2$.

A positive integer $p$ is a \emph{horizontal period} of $A$ if the $i$-th column of $A$ equals the $(i+p)$-th column of $A$ for all $i=1,\ldots,n-p$. We denote the smallest horizontal period of $A$ by $\hper(A)$.
Similarly, a positive integer $q$ is a \emph{vertical period} of $A$ if the $i$-th row of $A$ equals the $(i+q)$-th row of $A$ for all $i=1,\ldots,m-q$; the smallest vertical period of $A$ is denoted by $\vper(A)$. 

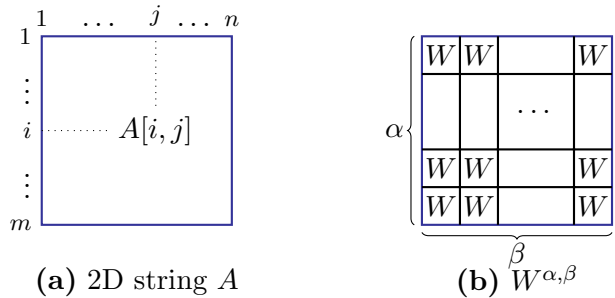
\begin{figure}[htpb]
    \centering
    \begin{tikzpicture}[scale=0.5]
    \begin{scope}[xshift=-10cm]
    \draw[thick,darkblue] (0,0) rectangle (5,5);
    \draw (0,5) node[left] {{\footnotesize 1}};
    \draw (0,3.75) node[left] {$\vdots$};
    \draw (0,2.5) node[left] {{\footnotesize $i$}};
    \draw (0,1.25) node[left] {$\vdots$};
    \draw (0,0) node[left] {{\footnotesize $m$}};
    \draw (5,5) node[above] {{\footnotesize $n$}};
    \draw (4,5) node[above] {$\ldots$};
    \draw (3,5) node[above] {{\footnotesize $j$}};
    \draw (1.5,5) node[above] {$\ldots$};
    \draw (0,5) node[above] {{\footnotesize 1}};
    \draw (3,2.5) node (A) {$A[i,j]$};
    \draw[dotted] (0,2.5) -- (A) -- (3,5);
    \draw (2.5,-1.5) node {\textbf{(a)} 2D string $A$};
    \end{scope}
    \draw[thick,darkblue] (0,0) rectangle (5,5);
    \foreach \x in {0,1,4}
      \foreach \y in {0,1,4}
        \draw[xshift=0.5cm,yshift=0.5cm] (\x,\y) node {$W$};
    \draw[thick] (1,0) -- (1,5) (2,0) -- (2,5) (4,0) -- (4,5) (0,1) -- (5,1) (0,2) -- (5,2) (0,4) -- (5,4);
    \draw (3,3) node {\ldots};
    \draw[snake=brace,yshift=-0.2cm] (5,0) -- node[below] {$\beta$} (0,0);
    \draw[snake=brace,xshift=-0.2cm] (0,0) -- node[left] {$\alpha$} (0,5);
    \draw (2.5,-1.5) node {\textbf{(b)} $W^{\alpha,\beta}$};
    \end{tikzpicture}
    \caption{A 2D-string and the structure of $W^{\alpha,\beta}$.}
    \label{fig:power}
\end{figure}

An $r \times c$ subarray $B=A[i_1 \dd i_2, j_1 \dd j_2]$ of $A$ is a \emph{2D-run} if $\hper(B)\leq c/2$, $\vper(B)\leq r/2$ and extending $B$ by a row or column, i.e.~either of $A[i_1-1, j_1 \dd j_2]$, $A[i_2+1, j_1 \dd j_2]$, $A[i_1 \dd i_2, j_1-1]$, or $A[i_1 \dd i_2, j_2+1]$, would result in a change of the smallest vertical or the horizontal period.

 If $W$ is a 2D array, then by $W^{\alpha,\beta}$ we denote an array that is composed of $\alpha \times \beta$ copies of $W$; see \cref{fig:power}(b). 
 A \emph{tandem} of $W$ is an array of the form $W^{1,2}$ and a \emph{quartic} of $W$ is the array $W^{2,2}$.

A 2D array $A$ is called \emph{primitive} if $A=B^{\alpha,\beta}$ for positive integers $\alpha,\beta$ implies that $\alpha=\beta=1$. The \emph{primitive root} of an array $A$ is the unique primitive array $B$ for which $A=B^{\alpha,\beta}$ for $\alpha,\beta \ge 1$.

Apostolico and Brimkov~\cite{DBLP:journals/tcs/ApostolicoB00} proved the following upper bound, and showed that it is tight by giving a corresponding lower bound.
\begin{fact}[Lemma 5 in~\cite{DBLP:journals/tcs/ApostolicoB00}]\label{fct:ABr}
A 2D array of size $n \times n$ has $\cO(n^2 \log^2 n)$ occurrences of primitively rooted quartics.
\end{fact}

We say that a quartic $Q=W^{2,2}$ is \emph{induced} by a 2D-run $R$ if $Q$ is a subarray of $R$ and $\hper(R)$ and $\vper(R)$ divide the width and height of $W$, respectively.

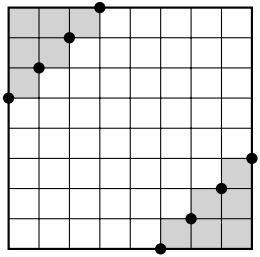
\begin{figure}[htpb]
    \begin{center}

\begin{tikzpicture}[scale=0.4]
\foreach \xa/\xb/\y in {0/2/7, 0/1/6, 0/0/5, 7/7/2, 6/7/1, 5/7/0} {
  \foreach \x in {\xa,...,\xb} {
    \filldraw[color=gray!35] (\x,\y) rectangle (\x+1,\y+1); 
  }
}

\draw[step=1.0,black] (0,0) grid (8, 8);
\draw[thick] (0,0) rectangle (8, 8);

\foreach \x/\y in {3/8, 2/7, 1/6, 0/5, 8/3, 7/2, 6/1, 5/0} {
  \node[circle, fill,inner sep=1.5pt] at (\x,\y) {};
}
\end{tikzpicture}
\end{center}
    \caption{Shaded positions contain letters $b$, all the other the letters $a$. Each rectangle with top-left and bottom-right corners marked is a 2D-run;
    altogether there are 18 distinct 2D-runs, including two of the form $b^{2,2}$. There are also 10 distinct quartics $a^{\alpha,\beta}$,
    where $0<\alpha,\beta\le 8$ are even and $\alpha+\beta\le 10$. 
    There is also the quartic $b^{2,2}$ (altogether 11 distinct quartics). The centrally
    placed quartic $a^{2,2}$ is contained in 16 2D-runs.
    There are only two distinct primitively rooted quartics.}
    \label{fig:walen}
\end{figure}

\begin{observation}\label{obs:ind_qu}
Every quartic is induced by a 2D-run. However; the same quartic can be induced even by $\Theta(n^2)$ 2D-runs; say the middle quartic $a^{2,2}$ in \cref{fig:walen}.
\end{observation}

\begin{remark}
The fact that a string of length $n$ has $\cO(n \log n)$ occurrences of primitively rooted squares immediately shows (by the fact that a square is induced by exactly one run) that it has $\cO(n \log n)$ runs. However, an analogous argument applied for quartics and 2D-runs does not give a non-trivial upper bound for the number of the latter because of Observation~\ref{obs:ind_qu}.
\end{remark} 

In our algorithms, we use a variant of the Dictionary of Basic Factors in 2D (2D-DBF in short) that is similar to the one presented in~\cite{DBLP:books/daglib/0020111}. Namely, to each subarray of $A$ whose width and height is an integer power of 2 we assign an integer identifier from $[0, n^2]$ so that two arrays with the same dimensions are equal if and only if their identifiers are equal. The total number of such subarrays is $\cO(n^2 \log^2 n)$ and the identifiers can be assigned in $\cO(n^2 \log^2 n)$ time; see~\cite{DBLP:books/daglib/0020111}. Using 2D-DBF, we can assign an identifier to a subarray of $A$ of arbitrary dimensions $r \times c$ being a quadruple of 2D-DBF identifiers of its four $2^i \times 2^j$ subarrays that share one of its corners, where $2^i \le r < 2^{i+1}$ and $2^j \le c < 2^{j+1}$. Such quadruples preserve the property that two subarrays of the same dimensions are equal if and only if the 2D-DBF quadruples are the same.

As an illustration, we show a tight bound for the number of distinct tandems and an optimal algorithm for computing them.
\begin{theorem}\label{lem:n3_tan}
The maximum number of distinct tandems in an $n \times n$ array $A$ is $\Theta(n^3)$. All distinct tandems in an $n \times n$ array can be reported in the optimal $\Theta(n^3)$ time.
\end{theorem}
\begin{proof}
Let us fix two row numbers $i<i'$ in $A$. Then, the number of distinct tandems with top row $i$ and bottom row $i'$ is $\cO(n)$ by the fact that a string of length $n$ contains $\cO(n)$ squares~\cite{DBLP:journals/dam/DezaFT15,DBLP:journals/jct/FraenkelS98}. Thus, in total there are $\cO(n^3)$ distinct tandems.
For the lower bound, let the $i$th row of $A$ be filled with occurrences of the letter $i$. Every subarray of $A$ of even width is a tandem.
For each distinct triplet of top and bottom rows and even width, we obtain a distinct tandem.

Let us proceed to the algorithm.
For a height $h \in [1,n]$, we assign integer identifiers from $[1,n^2]$ that preserve lexicographical comparison to all height-$h$ substrings of columns of $A$. They can be assigned using the generalized suffix tree~\cite{DBLP:books/daglib/0020103,DBLP:journals/algorithmica/Ukkonen95} of the columns of $A$ in $\cO(n^2 \log n)$ time.
Let $B_h$ be an array such that $B_h[i,j]$ stores the identifier of $A[i \dd i+h-1,j]$. To a subarray $W=A[i \dd i+h-1,j \dd j+w-1]$ we assign an \emph{identifier} $\iden(W)=B_h[i,j \dd j+w-1]$.
Then for any two subarrays $W$ and $W'$ of height $h$, $W=W'$ if and only if $\iden(W)=\iden(W')$.
For every height $h=1,\ldots,n$ and row $i$, we find all distinct squares in $B_h[i,1],\ldots,B_h[i,n]$ in $\cO(n)$ time~\cite{DBLP:conf/cpm/BannaiIK17,DBLP:journals/tcs/CrochemoreIKRRW14,DBLP:journals/jcss/GusfieldS04}. This corresponds to the set of distinct tandems with top row $i$ and bottom row $i+h-1$.
Finally, we assign identifiers from 2D-DBF of $A$ to each of the tandems and use radix sort to sort them and enumerate distinct tandems.
\end{proof}

\section{Improved Upper Bound for 2D-Runs}\label{sec:runs}

We introduce the framework that Amir et al.\ used for efficiently computing 2D-runs~\cite{DBLP:conf/esa/AmirLMS18,Amir2020}.

We say that a subarray $B=A[i_1 \dd i_2, j_1 \dd j_2]$ of $A$ is a \emph{horizontal run} if it is horizontally periodic (that is, $\hper(B)\leq \width(B)/2$) and extending $B$ by either of the columns  $A[i_1 \dd i_2, j_1-1]$ or $A[i_1 \dd i_2, j_2+1]$ would result in a change of the smallest horizontal period. 
(Note that $B$ does not have to be vertically periodic.)

For $k \in [1, \lfloor \log n \rfloor]$ and $i \in [1, n-2^k+1]$, let $H^k_i$ be the string obtained by replacing the columns of array $A[i \dd i+2^k-1, 1 \dd n]$ with metasymbols such that $H^k_i[j]=H^k_i[j']$ if and only if $A[i \dd i+2^k-1, j]=A[i \dd i+2^k-1, j']$. Notice that each such horizontal run of height $2^k$ corresponds to a run in some $H_i^k$.

The following lemma will enable us to ``anchor'' each 2D-run $R$ in the top-left or bottom-left corner of a horizontal run of ``similar'' height as $R$.
It was proved in~\cite{Amir2020}, but we provide a proof for completeness.

\begin{lemma}[Lemma 7 in \cite{Amir2020}]\label{lem:align}
Let $R$ be a 2D-run whose height is in the range $[2^k,2^{k+1})$.
Then there is a horizontal run $R'$ of height $2^k$ with  $\hper(R')=\hper(R)$ and $\width(R')\ge \width(R)$
such that top-left or bottom-left corners of $R$ and $R'$ coincide (see \cref{fig:hor_run}). 
\end{lemma}
\begin{proof}
Let $R=A[i_1 \dd i_2, j_1 \dd j_2]$ be the 2D-run in scope and let $k=\lfloor \log (i_2-i_1+1) \rfloor$.
We have to show that at least one of the two following statements holds.
\begin{itemize}
\item There is a run $R_1=S[j_1 \dd b]$ in $S=H^k_{i_1}$ with smallest period $p$ and $b \geq j_2$.
\item There is a run $R_2=T[j_1 \dd d]$ in $T=H^k_{i_2-2^k+1}$ with smallest period $p$ and $d \geq j_2$.
\end{itemize}

Since $\vper(R)\leq \height(R)/2$, all distinct rows of $R$ are represented in each of $U=S[j_1 \dd j_2]$ and $V=T[j_1 \dd j_2]$ and hence $p=\per(U)=\per(V)$. 
Let $R_1=S[a \dd b]$ be the run that extends $U$ and $R_2=T[c \dd d]$ be the run that extends $V$.
Let us suppose towards a contradiction that $\max (a,c)<j_1$. Then, $A[i_1 \dd i_2, j_1-1]=A[i_1 \dd i_2, j_1-1+p]$, which contradicts $R$ being a run, since $R$ and $B=A[i_1 \dd i_2, j_1-1 \dd j_2]$ have the same horizontal and vertical periods.
\end{proof}

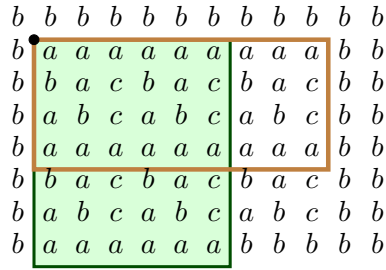
\begin{figure}[htpb]
    \centering
\begin{tikzpicture}[scale=0.43]

\filldraw[xshift=-0.5cm,yshift=-0.1cm, green!15!white] (1,8) rectangle (7,1);

\draw[xshift=-0.5cm,yshift=-0.1cm, green!30!black,line width=0.4mm] (1,8) rectangle (7,1); 

\draw[xshift=-0.5cm,yshift=-0.1cm, brown,line width=0.6mm,] (1,8) rectangle (10,4); 

\filldraw[xshift=-0.5cm,yshift=-0.1cm] (1,8) circle (0.15cm);

\foreach \dx in {1,4,7}{
    \foreach \dy in {2,5}{
        \foreach \x/\y/\c in {0/0/a, 1/0/b, 2/0/c, 0/2/a, 1/2/a, 2/2/a, 0/1/b, 1/1/a, 2/1/c}{
            \draw[xshift=\dx cm,yshift=\dy cm] (\x,\y) node[above] {$\c$};
        }
    }
}

\foreach \x in {1,2,3,4,5,6}{
    \draw (\x,1) node[above] {$a$};
}


\foreach \dx in {0,2,4,6,8,10}{
    \foreach \dy in {8}{
        \foreach \x/\y/\c in {0/0/b, 1/0/b}{
            \draw[xshift=\dx cm,yshift=\dy cm] (\x,\y) node[above] {$\c$};
        }
    }
}

\foreach \dx in {0,10,11}{
    \foreach \dy in {1,2,3,4,5,6,7}{
            \draw[xshift=\dx cm,yshift=\dy cm] (0,0) node[above] {$b$};
    }
}

\foreach \dx in {7,8,9}{
            \draw[xshift=\dx cm,yshift=1 cm] (0,0) node[above] {$b$};
}
\end{tikzpicture}
\caption{The shaded $7\times 6$ subarray is a 2D-run $R$, with vertical period $3$ and horizontal period $p=3$.
The other marked $4\times 9$ rectangle encloses a horizontal run $R'$
with the same top-left corner and the same horizontal period as $R$. We have $2 \cdot p\le width(R)\le width(R')$.
}
    \label{fig:hor_run}
\end{figure}

The sum of the lengths of the runs in a string of length $n$ can be $\Omega(n^2)$ as 
shown in~\cite{DBLP:journals/tcs/GlenS13}.
However, we prove the following lemma, which is crucial for our approach.
We will use it to obtain an overall bound on the possible widths of 2D-runs for our anchors.

\begin{lemma}\label{lem:1d}
For any string $S$ of length $n$ we have that
\[\rho(S)\;:=\; \sum_{R \in \mathcal{R}(S)} (|R| - 2\cdot \per(R)+1) \;= \; \cO(n \log n).\]
\end{lemma}
\begin{proof}
We consider for each run $R=S[i \dd j]$ of $S$ the interval $I_R=[i,j-2\cdot \per(R)+1]$. 
Note that $\rho(S)=\sum_{R \in \mathcal{R}(S)} |I_R|$.

Observe that for every $a \in I_R$ the string $S[a \dd a+\per(R)-1]$ is primitive, since if it was of the form $U^k$ for a string $U$ and an integer $k>1$, then $|U|<\per(R)$ would be a period of $R$, a contradiction.
Hence, at each position $a\in I_R$ there is an occurrence of a primitively rooted square of length $2\cdot \per(R)$.

A direct application of the Three Squares Lemma (\cref{lem:3sq}) implies that at most $\cO(\log n)$ primitively rooted squares can start at each position $a$.
Each such square extends to a unique run.
Thus, each position $i$ belongs to $\cO(\log n)$ intervals $I_R$ for $R \in \mathcal{R}(S)$.
This completes the proof.
\end{proof}

We are now ready to prove the main result of this section.

\begin{theorem}
\label{thm:main}
There are $\cO(n^2 \log^2 n)$ 2D-runs in an $n \times n$ array $A$.
\end{theorem}
\begin{proof}
We will iterate over all horizontal runs $R'=A[i \dd i',j \dd j']$ whose height is a power of $2$, i.e.~$i'=i+2^k-1$ for some $k$.
For each such horizontal run $R'$, we consider the 2D-runs $R$ with:
\begin{enumerate}[(a)]
    \item top-left corner $A[i,j]$ or bottom-left corner $A[i',j]$,
    \item $\hper(R)=\hper(R')$, and\label{it:hper} 
    \item $\height(R) \in [2^k,2^{k+1})$.
\end{enumerate}
For each such 2D-run $R$, we have $\width(R) \in [2 \cdot \hper(R'),\width(R')]$, else the horizontal period would break, i.e.~property~\eqref{it:hper} would be violated.
Let us notice that $R'$ corresponds to a run $U=H^k_i[j \dd j'] \in \mathcal{R} (H^k_i)$.
In particular, $\width(R)\in [2 \cdot \per(U),|U|]$.

\cref{lem:align} implies that each 2D-run is accounted for at least once in this manner.
It is thus enough to bound the number of considered runs.
We have $n$ choices for $i$ and $\log n$ choices for $k$.
Further, due to~\cref{lem:1d}, for each corresponding meta-string $H^k_i$ we have $\cO(n \log n)$ choices for a pair $(j,c)$ such that $U=H^k_i[j \dd j'] \in \mathcal{R} (H^k_i)$ and $c \in [2 \cdot \per(U),|U|]$.
In total, we thus have $\cO(n^2 \log^2 n)$ choices for $(i,k,j,c)$.
We will complete the proof by showing that there is only a constant number of 2D-runs with top-left corner $A[i,j]$, width $w$ and whose height is in the range $[2^k,2^{k+1})$.
(2D-runs with bottom-left corner $A[i',j]$ can be bounded symmetrically.)

\begin{claim}[cf.\ Lemma 10 in \cite{Amir2020}]\label{lem:3vertical}
Let $B$ be an $r \times c$ array with $r \in [2^k, 2^{k+1})$. Then, there
are at most two integers $p>2^{k-1}$ such that $p=\vper(B') \leq \height(B')/2$ for $B'$ consisting of the top $\height(B')\ge 2^k$ rows of $B$.
\end{claim}
\begin{proof}
Consider $S$ to be the meta-string obtained by replacing the rows of $B$ by single letters.
Then, a direct application of the Three Squares Lemma  (\cref{lem:3sq}) to $S$ yields the claimed bound.
\end{proof}

We apply \cref{lem:3vertical} to $B=A[i\dd \min(i+2^{k+1}-2,n),j\dd j+c-1]$.
If $\vper(R) \le 2^{k-1}$, then $\vper(R) = \vper(R')$ by the Periodicity Lemma (Lemma~\ref{lem:FW}) applied to the meta-string obtained by replacing the rows of the intersection of $R'$ and $B$ by single letters.
Now \cref{lem:3vertical} implies that there are at most three choices to make for the vertical period: $\vper(R')$ and the two integers from the claim. 
Finally, for fixed top-left corner, width and vertical period we can have a single 2D-run.
This concludes the proof.
\end{proof}

Amir et al.~\cite{Amir2020} presented the following algorithmic result.

\begin{theorem}[\cite{Amir2020}]\label{thm:ami}
All 2D-runs in an $n \times n$ array can be computed in $\cO(n^2 \log n + \textsf{output})$ time, where $\textsf{output}$ is the number of 2D-runs reported.
\end{theorem}

By combining~\cref{thm:main,thm:ami} we get the following corollary.

\begin{corollary}\label{cor:runs}
All 2D-runs in an $n \times n$ array can be computed in $\cO(n^2 \log^2 n)$ time.
\end{corollary}

\section{Upper Bound on the Number of Distinct Quartics}\label{sec:quartics}
Fact~\ref{fct:ABr} that originates from~\cite{DBLP:journals/tcs/ApostolicoB00} shows that an $n \times n$ array $A$ has $\cO(n^2 \log^2n)$ occurrences of primitively rooted quartics. This obviously implies that the number of distinct primitively rooted quartics is upper bounded by $\cO(n^2 \log^2n)$. Unfortunately, an array can contain $\Theta(n^4)$ occurrences of general quartics; this takes place e.g.~for a unary array. In this section we show that $\cO(n^2 \log^2n)$ is also an upper bound for the number of \emph{distinct} general quartics, i.e.~subarrays of $A$ of the form $W^{\alpha,\beta}$ for even $\alpha,\beta\ge 2$ and primitive $W$.

The following lemma and its corollary are the combinatorial foundation of our proofs. An array $W$ with $\height(W) \in [2^a,2^{a+1})$ and $\width(W) \in [2^b,2^{b+1})$ will be called an \emph{$(a,b)$-array}.

\begin{lemma}\label{lem:per2D}
Let $a,b$ be non-negative integers and $W,W'$ be different primitive $(a,b)$-arrays. If occurrences of $W^{2,3}$ and $(W')^{2,3}$ (of $W^{3,2}$ and $(W')^{3,2}$, respectively) in $A$ share the same corner (i.e., top-left, top-right, bottom-left or bottom-right), then $\width(W)=\width(W')$ ($\height(W)=\height(W')$, respectively).
\end{lemma}
\begin{proof}
Clearly it is sufficient to prove the lemma for $W^{2,3}$ and $(W')^{2,3}$. Assume w.l.o.g.\ that occurrences of $W^{2,3}$ and $(W')^{2,3}$ in $A$ share the top-left corner and consider their overlap $X$.

Each of the rows of $X$ has periods $\width(W)$ and $\width(W')$. Let us assume w.l.o.g.\ that $\width(W) \le \width(W')$.
Then
\[\width(X) = 3\cdot \width(W) \ge \width(W)+2^{a+1} \ge \width(W)+\width(W').\]
By the Periodicity Lemma (Lemma~\ref{lem:FW}), $p=\gcd(\width(W),\width(W'))$ is a horizontal period of $X$.

The array $X$ contains at least one occurrence of $W$ and $W'$ in its top-left corner. Hence, $W$ and $W'$ have a horizontal period $p$. If $\width(W) < \width(W')$, then $\width(W')$ cannot be a multiple of $\width(W)$, because then we would have $\width(W') > 2^{a+1}$. Hence, if $\width(W) < \width(W')$, we would have $p < \width(W)$ which by $p \mid \width(W)$ would mean that $W$ is not primitive. This indeed shows that $\width(W)=\width(W')$.
\end{proof}

\begin{corollary}\label{cor:per2D}
Let $a,b$ be non-negative integers and $W,W'$ be different $(a,b)$-arrays. If occurrences of $W^{3,3}$ and $(W')^{3,3}$ in $A$ share the same corner (i.e., top-left, top-right, bottom-left or bottom-right), then at least one of $W$, $W'$ is not primitive.
\end{corollary}

If $V^{2,2}$ is a non-primitively rooted quartic, then there exists a primitive array $W$ such that $V=W^{\alpha,\beta}$ and at least one of $\alpha,\beta$ is greater than one. We will call the quartic $W^{2\alpha,2\beta}$ \emph{thin} if $\alpha=1$ or $\beta=1$ for this decomposition, and \emph{thick} otherwise.
We refer to \emph{points} in $A$ as the $(n+1)^2$ positions where row and column delimiters intersect.
Let us first bound the number of distinct thin quartics. For $\beta>1$, we consider any rightmost occurrence of every such quartic, that is, any occurrence $A[i_1 \dd i_2,j_1 \dd j_2]$ that maximizes $j_1$.

\begin{lemma}\label{lem:qtr_thin}
The total number of distinct thin quartics in $A$ is $\cO(n^2 \log^2n)$.
\end{lemma}
\begin{proof}
We give a proof for quartics of the form $W^{2,2\beta}$ for primitive $W$ and $\beta>1$; the proof for quartics of the form $W^{2\alpha,2}$ for $\alpha>1$ is symmetric. We consider each pair of positive integers $a,b$ and show that each point holds the top-left corner of at most two rightmost occurrences of $W^{2,2\beta}$ for primitive $(a,b)$-arrays $W$ and $\beta>1$.

Assume to the contrary that the rightmost occurrences of $W^{2,2\beta}$, $(W')^{2,2\beta'}$ and $(W'')^{2,2\beta''}$ share their top-left corner for primitive $(a,b)$-arrays $W,W',W''$. The arrays $W,W',W''$ are pairwise different, since otherwise one of the occurrences would not be the rightmost. By Lemma~\ref{lem:per2D}, we have $\width(W)=\width(W')=\width(W'')$. Assume w.l.o.g.\ that $\height(W) < \height(W') < \height(W'')$.

Let $(i,j)$ denote the top-left corner of the three quartics. Let us consider three length-$2\ell$ strings formed of metacharacters that correspond to row fragments:
\[(A[i,j \dd j+w-1]),\ldots,(A[i+2\ell-1,j \dd j+w-1])\]
for $w=\width(W)$ and $\ell \in \{\height(W),\height(W'),\height(W'')\}$. All the three strings need to be primitively rooted squares. We apply the Three Squares Lemma (Lemma~\ref{lem:3sq}) to conclude that
$\height(W'') > \height(W)+\height(W') > 2^{a+1},$
a contradiction.
\end{proof}

Now let us proceed to thick quartics. Unfortunately, in this case a single point can be the top-left corner of a linear number of rightmost occurrences of thick quartics; see the example in \cref{fig:walen}.
Let us consider an occurrence of $W^{\alpha,\beta}$ for even $\alpha,\beta > 2$ and primitive $W$, called a \emph{positioned quartic}. It implies $\alpha\cdot \beta$ occurrences of $W$. Let us call all corners of all these occurrences of $W$ \emph{special points} of this positioned quartic. Each special point stores a direction in $\{\text{top-left},\text{top-right},$ $\text{bottom-left},\text{bottom-right}\}$. A special point has one of the directions if it is the respective corner of an occurrence of $W^{3,3}$ in this positioned quartic. Clearly, since $\alpha,\beta \ge 4$, for every special point in $W^{\alpha,\beta}$ except for the middle row if $\alpha=4$ or middle column if $\beta=4$, one can assign such a direction (if many directions are possible, we choose an arbitrary one); see \cref{fig:arrows}.

\begin{figure}[htpb]
\begin{center}

\begin{tikzpicture}[scale=0.55]

\foreach \xa/\xb/\y in {0/2/0, 0/2/1, 0/2/2} { \foreach \x in {\xa,...,\xb} {
  \draw[->,thick] (\x,\y)--+(0.4,0.4);
}}

\foreach \xa/\xb/\y in {3/6/0, 3/6/1, 3/6/2} { \foreach \x in {\xa,...,\xb} {
  \draw[->,thick] (\x,\y)--+(-0.4,0.4);
}}

\foreach \xa/\xb/\y in {0/2/3, 0/2/4, 0/2/5, 0/2/6} { \foreach \x in {\xa,...,\xb} {
  \draw[->,thick] (\x,\y)--+(0.4,-0.4);
}}

\foreach \xa/\xb/\y in {3/6/3, 3/6/4, 3/6/5, 3/6/6} { \foreach \x in {\xa,...,\xb} {
  \draw[->,thick] (\x,\y)--+(-0.4,-0.4);
}}

\draw[step=1] (0,0) grid (6,6);
\draw[thick,darkblue] (0,0) rectangle (6,6);
\draw (-1.3,0) node[above] {$W^{6,6}$};

\begin{scope}[xshift=10cm]
\foreach \xa/\xb/\y in {0/1/0, 0/1/1} { \foreach \x in {\xa,...,\xb} {
  \draw[->,thick] (\x,\y)--+(0.4,0.4);
}}

\foreach \xa/\xb/\y in {3/4/0, 3/4/1} { \foreach \x in {\xa,...,\xb} {
  \draw[->,thick] (\x,\y)--+(-0.4,0.4);
}}

\foreach \xa/\xb/\y in {0/1/3, 0/1/4} { \foreach \x in {\xa,...,\xb} {
  \draw[->,thick] (\x,\y)--+(0.4,-0.4);
}}

\foreach \xa/\xb/\y in {3/4/3, 3/4/4} { \foreach \x in {\xa,...,\xb} {
  \draw[->,thick] (\x,\y)--+(-0.4,-0.4);
}}

\draw[step=1] (0,0) grid (4,4);
\draw[thick,darkblue] (0,0) rectangle (4,4);
\draw (-1.3,0) node[above] {$W^{4,4}$};
\end{scope}

\end{tikzpicture}

\end{center}
\caption{Special points of a positioned quartic with primitive root $W$
with associated directions of four types. 
The arrow indicates the corner (four possibilities) of $W^{3,3}$
which is contained in the quartic. If several assignments of directions are
possible, only one of them is chosen (it does not matter which one).
In case of $W^{4,4}$ the middle row and column are not special.
}\label{fig:arrows}
\end{figure}
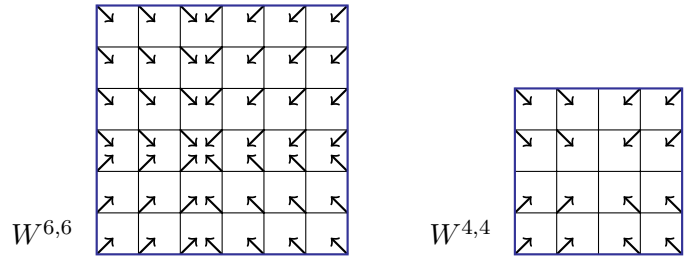

The quartics with primitive root $W$ are called \emph{$W$-quartics}.
The set of all special points (with directions) of {\it all} positioned
thick $W$-quartics for a given $W$ is denoted by $\SP(W)$. 
Among $W$-quartics of the same height we distinguish the ones with 
maximal width, which we call \emph{h-maximal} (horizontally maximal).
Let us observe that each $W$-quartic is contained in an occurrence of some h-maximal $W$-quartic.

\begin{theorem}\label{thm:dist_q}
The number of distinct quartics in an $n \times n$ array is $\cO(n^2 \log^2n)$.
\end{theorem}
\begin{proof}
By Fact~\ref{fct:ABr} and Lemma~\ref{lem:qtr_thin} it suffices to show that the total number of distinct thick quartics in $A$ is $\cO(n^2 \log^2n)$.
Let us fix non-negative integers $a,b$.
It is enough to show that the number of distinct subarrays of $A$ of the form $W^{\alpha,\beta}$ for even $\alpha,\beta > 2$ and any primitive $(a,b)$-array $W$ is $\cO(n^2)$.

The sets of special points have the following properties. \cref{cl:spsp} follows from Corollary~\ref{cor:per2D}.

\begin{claim}\label{cl:spsp}
For primitive $(a,b)$-arrays $W \ne W'$ ,
$\SP(W) \cap \SP(W')=\emptyset$.
\end{claim}

For an array $W$, let us denote by $\Thick(W)$ the total number of thick quartics in $A$ with primitive root $W$.

\begin{claim}
For a primitive $(a,b)$-array $W$, $\Thick(W) < |\SP(W)|$.
\end{claim}
\begin{proof}
For each $\alpha =4,6,\ldots$ in this order, we select one positioned h-maximal $W$-quartic $U_{\alpha}$ of height $\alpha\cdot \height(W)$. The number of distinct $W$-quartics in $A$ of 
height $\alpha\cdot \height(W)$ is at most the number of special points in $U_\alpha$ in 
any of its rows. Note that this statement also holds if $U_\alpha = W^{\alpha,4}$; then there are still four special points in each (non-middle if $\alpha=4$) row.

We describe a process of assigning distinct $W$-quartics to distinct special points in the set $\SP(W)$.
Assume all points in this set are initially not marked. We choose any   single row 
from $U_\alpha$
with all special points in this row still not marked. Then we mark all these special points.
We can always choose a suitable row because the heights are increasing.

This way each $W$-quartic is assigned to only one special point from $\SP(W)$.
\end{proof}

By the claims, the total number of thick $W$-quartics for primitive $(a,b)$-arrays $W$ is bounded by:
\[\sum_{W} \Thick(W) < \sum_W |\SP(W)| \le 4(n+1)^2,\]
where the sum is over all primitive $(a,b)$-arrays $W$. The conclusion follows.
\end{proof}

\section{Algorithms for Computing Quartics}\label{sec:apps}
In this section we show algorithmic applications of 2D-runs related to quartics.

\begin{theorem}\label{thm:occ_prim}
All occurrences of primitively rooted quartics in an $n \times n$ array $A$ can be computed in the optimal $\cO(n^2\log^2 n)$ time.
\end{theorem}
\begin{proof}
Let us consider a 2D-run $R=A[i_1 \dd i_2,j_1 \dd j_2]$ with periods $\hper(R)=p$ and $\vper(R)=q$. It induces primitively rooted quartics of width $2p$ and height $2q$. The set of top-left corners of these quartics forms a rectangle $\hat{R}=[i_1,i_2-2p+1]\times [j_1,j_2-2q+1]$. We denote by $\F_{p,q}$ the family of such rectangles $\hat{R}$ over 2D-runs $R$ with the same periods $p,q$.

Such rectangles for different 2D-runs may overlap, even when the dimensions of the quartic are fixed (see Observation~\ref{obs:ind_qu}). In order not to report the same occurrence multiple times, we need to compute, for every dimensions of a quartic, all points in the union of the corresponding rectangles. 
This could be done with an additional $\log n$-factor in the complexity using a standard line sweep algorithm~\cite{Bentley}. However, we can achieve $\Oh(n^2 \log^2 n)$ total time using the fact that the total number of occurrences reported is $\cO(n^2\log^2 n)$.

\begin{claim}\label{clm:KleeReport}
Let $\F_1,\ldots,\F_k$ be families of 2D rectangles in $[1,n]^2$ and let $r=\sum_{i=1}^k |\F_i|$. We
can compute $k$ (not necessarily disjoint) sets of grid points $\mathit{Out}_i=\bigcup \F_i$  in $\cO(n+r+\textsf{output})$ total time, where $\textsf{output}=\sum_i\,|\mathit{Out}_i|$ is the total number of reported points.
\end{claim}
\begin{proof}
We design an efficient line sweep algorithm. We will perform a separate line sweep, left to right, for each family $\F_i$.

The sweep goes over horizontal ($x$) coordinates in a left-to-right manner. The broom stores vertical ($y$) coordinates of horizontal sides of rectangles that it currently intersects. They are stored in a sorted list $L$ of pairs $(y,c)$, where $y$ is the coordinate, and $c$ is the count of rectangles with bottom side at coordinate $y$ minus the count of the rectangles with top side at coordinate $y$. Only pairs with non-zero second component are stored. Clearly, the second components of the list elements always sum up to $0$.

A coordinate  $x$ is processed if $L$ is non-empty before accessing it or there exist any vertical sides of rectangles at $x$. All vertical sides with the same $y$-coordinate are processed in a batch. For every such batch we want to guarantee that endpoints of all sides are stored in a list $B$ in a top-down order.

A top (bottom) endpoint at vertical coordinate $y$ is stored as $(y,+1)$ ($(y,-1)$, respectively).

Let us now describe how to process a horizontal coordinate $x$. Let us merge the list $L$ that is currently in the broom with the list $B$ of the batch by the first components. If there is more than one pair with the same first component, we merge all of them together, summing up the second components. 
Let us denote by $L'$ the resulting list. We iterate over all elements of $L'$, keeping track of the partial sum of second components, denoted as $s$. For every element $(y,c)$ of $L'$, the point $(x,y)$ is reported for $\bigcup \F_i$. Moreover, if the partial sum $s$ before considering $c$ was positive and the previous element of $L'$ is $(y',c')$, all points $(x,y'+1),\ldots,(x,y-1)$ are reported to $\mathit{Out}_i$. 
Finally, all pairs with second component equal to zero are removed from $L'$ which becomes the new list $L$.

Let us now analyze the complexity of the algorithm. The line sweep makes $n$ steps. The total size of lists $B$ across all families $\F_i$ is $\cO(r)$ and they can be constructed simultaneously in $\cO(n+r)$ time via bucket sort. 
Processing a batch with list $B$ takes $\cO(|L|+|B|)$ time plus the time to report points in $\mathit{Out}_i$. As we have already noticed, the sum of $\cO(|B|)$ components is $\cO(r)$. For every element $(y,c)$ of the initial list $L$, a point with the vertical coordinate $y$ is reported upon merging; hence, the sum of $\cO(|L|)$ components is dominated by $\cO(\textsf{output})$. Overall we achieve time complexity $\cO(n+r+\textsf{output})$.
\end{proof}

We apply~\cref{clm:KleeReport} to the families $\F_{p,q}$. Then $r$ and $\textsf{output}$ are upper bounded by $\cO(n^2 \log^2 n)$ by \cref{thm:main} and \cref{fct:ABr}, respectively. 
The optimality of our algorithm's complexity is due to the $\Omega(n^2 \log^2 n)$ lower bound on the maximum number of occurrences of primitively rooted quartics from~\cite{DBLP:journals/tcs/ApostolicoB00}.
\end{proof}

We proceed to an efficient algorithm for enumerating distinct, not necessarily primitively rooted, quartics using 2D-runs. The solution for an analogous problem for 1-dimensional strings (computing distinct squares from runs) uses Lyndon roots of runs~\cite{DBLP:journals/tcs/CrochemoreIKRRW14}. However, in 2 dimensions it is not clear if a similar approach could be applied efficiently, say, with the aid of 2D Lyndon words~\cite{DBLP:journals/algorithmica/MarcusS17} as Lyndon roots of 2D-runs. We develop a different approach in which the workhorse is the following auxiliary problem related to the folklore nearest smaller value problem.

Let us consider a grid of height $m$ in which every cell is either black or white. We say that the grid forms a 
\emph{staircase} if the set of white cells in each row is nonempty and is a prefix of this row
(see \cref{fig:staircase}).
A staircase can be uniquely determined by an array $\Whites[1 \dd m]$ such that $\Whites[i]$ is the number of white cells in the $i$th row. 
We consider \emph{shapes} of white rectangles. Each shape is a pair $(p,q)$ that represents the dimensions of the rectangle.
These shapes (and the corresponding rectangles) are partially ordered by the relation:
\[(p,q) < (p',q')\ \Leftrightarrow \ (p,q)\ne (p',q')\ \land\  p\le p'\ \land\ q\le q'.\]
Let us now consider the following problem.

\defproblem{Max White Rectangles}{An array $\Whites[1 \dd m]$ that represents a staircase.}{Shapes of all maximal white rectangles in this staircase.}

\begin{lemma}
\label{lem:MWR}
\textsc{Max White Rectangles} problem can be solved in $\cO(m)$ time.
\end{lemma}
\newcommand{\NSVUp}{\mathit{NSVUp}}
\newcommand{\NSVDown}{\mathit{NSVDown}}
\newcommand{\MaxWidth}{\mathit{MaxWidth}}
\newcommand{\mw}{\mathit{mw}}

\begin{proof}
Assume that $\Whites[0]=\Whites[m+1]=-1$. Let us define two tables of size $m$:
\begin{align*}\NSVUp[i] &= \max\{ j\, :\, j < i,\, \Whites[j] < \Whites[i] \},\\
\NSVDown[i] &= \min\{ j\, :\, j > i,\, \Whites[j] < \Whites[i] \}.
\end{align*}

They can be computed in $\cO(m)$ time by a folklore algorithm for the nearest smaller value table; see e.g.~\cite{DBLP:journals/jal/BerkmanSV93}.
Then the problem can be solved as in Algorithm 1 presented below. 
After the first for-loop, 
for each maximal white rectangle $R$ we have $\MaxWidth[\height(R)]=\width(R)$, but we 
could have redundant values for non-maximal rectangles. In order to filter out non-maximal rectangles, we process the candidates by decreasing height and remove
the ones that are dominated by the previous maximal rectangle in the partial order of shapes.\end{proof}

\begin{algorithm}[htpb]
\DontPrintSemicolon
{\bf ComputeCandidates:}\;\vspace*{0.1cm}
  $\MaxWidth[1 \dd m]:=(0,\ldots,0)$\;
  \For{$i:= 1$ \KwSty{to} $m$}{
    $h:=\NSVDown[i] - \NSVUp[i] - 1$\;
    $\MaxWidth[h]:=\max(\MaxWidth[h],\Whites[i])$\;
  }
  \vspace*{0.1cm}
{\bf ReportMaximal:}\;
  $\mw:=0$\;
  \For{$h:=m$ \KwSty{down\ to} $1$}{
    \If{$\MaxWidth[h]>\mw$}{
      Report the shape $(h,\MaxWidth[h])$\;
      $\mw:=\MaxWidth[h]$\;
    }
  }
  \;
\caption{The first phase  computes a set of shapes of type  $(h,\MaxWidth[h])$, at most one for each height $h$; see also \cref{fig:staircase}.
In the second phase only inclusion-maximal shapes from this set are reported.
}\label{alg:NSV2}
\end{algorithm}

\begin{figure}[htpb!]
    \centering
\begin{tikzpicture}[scale=0.39]
    \filldraw[black!20!white] (0,1) rectangle (9,8);
    \foreach \y/\x in {1/3,2/7,3/6,4/7,5/8,6/6,7/2}{
        \filldraw[white] (0,\y) rectangle (\x,\y+1);
    }
    \foreach \y/\x in {1/3,2/7,3/6,4/7,5/8,6/6,7/2}{
        \draw (0,\y) grid (\x,\y+1);
    }
    \draw[very thick,black](0,2) rectangle (6,7);
    \draw (0,3.5) node[left] (i) {$i$};
    \draw[-latex] (i) .. controls (-1.5,4.5) and (-1.5,6.5) .. node[left] {$\NSVUp[i]$} (-0.3,7.5);
    \draw[-latex] (i) .. controls (-1,3) and (-1,2) .. node[left] {$\NSVDown[i]$} (-0.3,1.5);
    \draw (0,1) rectangle (9,8);
\end{tikzpicture}
    \caption{A maximal white rectangle containing row $i$ is computed using the NSV tables for $i$.}
    \label{fig:staircase}
\end{figure}
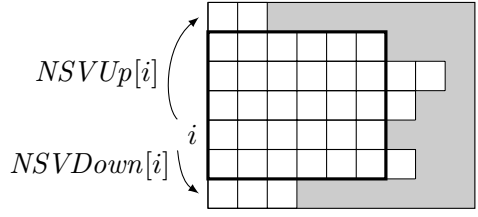

\begin{remark}
Note that the total area (and width) of a staircase can be large but the complexity of our algorithm is linear with respect to the number of rows, thanks to the small representation (array $\Whites$).
\end{remark}

Now our approach is graph-theoretic. The graph nodes correspond to occurrences of primitively rooted quartics.
For a fixed primitively rooted quartic $W^{2,2}$ we consider the graph $G_W=(V,E)$, where $V$ is the set of top-left corners of occurrences of $W^{2,2}$. Let $r=\height(W)$ and $c=\width(W)$. The edges in $G$ connect vertex $(i,j)$ with vertices $(i\pm r,j)$ and $(i,j\pm c)$, if they exist. See also \cref{fig:abc}.
This graph can be efficiently computed since we know its nodes due to Theorem~\ref{thm:occ_prim}. 

\begin{figure}[htpb]
    \centering
\begin{tikzpicture}[scale=0.45]
\begin{scope}[xshift=-5cm,yshift=3cm]
\foreach \x/\y/\c in {0/0/b,0/1/c,0/2/a, 1/0/c,1/1/a,1/2/b, 2/0/a,2/1/b,2/2/c}{
    \draw(\x,\y) node[above] {$\c$};
}
\draw[xshift=-0.5cm,yshift=-0.1cm,] (0,0) rectangle (3,3);
\draw (1,-0.1) node[below] {$W$};
\end{scope}
\begin{scope}
\clip (-0.5,0) rectangle (13.5,10);
\foreach \dx in {0,3,6,9,12}{
    \foreach \dy in {0,3,6}{
        \foreach \x/\y/\c in {0/0/b,0/1/c,0/2/a, 1/0/c,1/1/a,1/2/b, 2/0/a,2/1/b,2/2/c}{
            \draw[xshift=\dx cm,yshift=\dy cm] (\x,\y) node[above] {$\c$};
        }
    }
}
\end{scope}
\draw[xshift=-0.5cm,yshift=-0.1cm,step=3cm,blue] (0,0) grid (12,9);
\draw[xshift=0.5cm,yshift=1.9cm,step=3cm,thick,brown] (0,0) grid (12,6);
\draw[xshift=1.5cm,yshift=0.9cm,step=3cm,thick,green!50!black] (0,0) grid (12,6);
\foreach \x/\y in {0/9,3/9,6/9, 0/6,3/6,6/6}{
    \filldraw[xshift=-0.5cm,yshift=-0.1cm,blue] (\x,\y) circle (0.15cm);
}
\foreach \x/\y in {1/8,4/8,7/8}{
    \filldraw[xshift=-0.5cm,yshift=-0.1cm,brown] (\x,\y) circle (0.15cm);
}
\foreach \x/\y in {2/7,5/7,8/7}{
    \filldraw[xshift=-0.5cm,yshift=-0.1cm,green!50!black] (\x,\y) circle (0.15cm);
}
\draw (-1.5,0.5) node {$A$};
\begin{scope}[xshift=18cm,yshift=-0.5cm]
\draw (3,0.5) node {$G_W$};
\draw[blue] (0,6) rectangle (6,9)  (3,9) -- (3,6);
\foreach \x/\y in {0/9,3/9,6/9, 0/6,3/6,6/6}{
    \filldraw[blue] (\x,\y) circle (0.15cm);
}
\begin{scope}[yshift=-4cm,xshift=-1cm]
\draw[orange] (1,8) -- (7,8);
\foreach \x/\y in {1/8,4/8,7/8}{
    \filldraw[brown] (\x,\y) circle (0.15cm);
}
\end{scope}
\begin{scope}[yshift=-5cm,xshift=-2cm]
\draw[green!50!black] (2,7) -- (8,7);
\foreach \x/\y in {2/7,5/7,8/7}{
    \filldraw[green!50!black] (\x,\y) circle (0.15cm);
}
\end{scope}
\end{scope}
\end{tikzpicture}
    \caption{Graph $G_W$ has 12 vertices that form two components with 3 vertices each (green and brown) and one component with 6 vertices (blue). Note the non-trivial occurrences of $W$ in $W^{3,4}$.
}
    \label{fig:abc}
\end{figure}
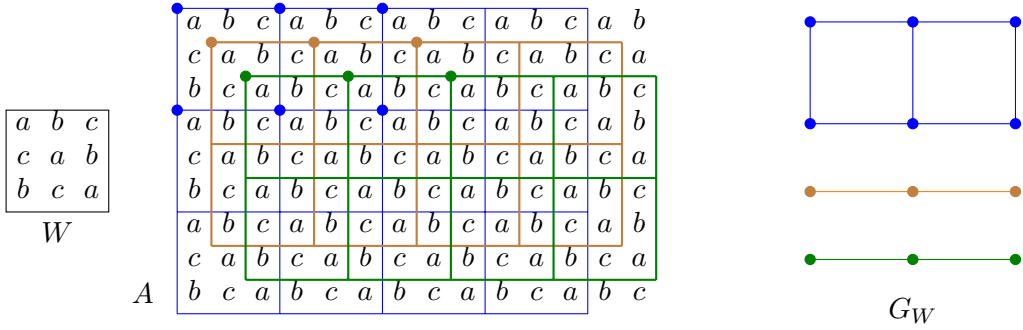

\begin{lemma}\label{lem:GW}
All graphs $G_W$, and their connected components, for all $W$ which are
primitive roots of quartics in $A$ can be constructed in $\cO(n^2\log^2n)$ time.
\end{lemma}
\begin{proof}
We first compute all occurrences of primitively rooted quartics in $A$ using Theorem~\ref{thm:occ_prim}. By Fact~\ref{fct:ABr}, there are $\cO(n^2 \log^2 n)$ of them in total. 

We can assign 2D-DBF identifiers (quadruples) to each of the occurrences and group the occurrences by distinct primitively rooted quartics via radix sort in $\cO(n^2\log^2 n)$ time. This gives us the vertices of $G_W$.

To compute the edges, we use an auxiliary $n \times n$ Boolean array $D$ that will store top-left corners of occurrences of each subsequent primitively rooted quartic $W^{2,2}$. 

Initially $D$ is set to zeroes and after each $W$, all cells with ones are zeroed in $\cO(|G_W|)$ time. Using this array and the positions of occurrences of $W^{2,2}$, the edges of $G_W$ can be computed in $\cO(|G_W|)$ time. It also allows to divide $G_W$ into connected components via graph search in $\cO(|G_W|)$ time.
\end{proof}

\begin{theorem}
All distinct quartics in an $n \times n$ array $A$ can be computed in $\cO(n^2\log^2 n)$ time.
\end{theorem}
\begin{proof}
We first apply~\cref{lem:GW}. Now consider a fixed primitive $W$ of height $c$ and width $r$.
Let us note that if $(i,j),(i',j')$ belong to the same connected component $H$ of $G_W$, then $i \equiv i' \pmod{r}$ and $j \equiv j' \pmod{c}$. We say that a connected component $H$ of $G_W$ \emph{generates} an occurrence of a power $W^{\alpha,\beta}$ if the $\alpha \beta$ occurrences of $W$ that are implied by it belong to $H$. If $W^{\alpha,\beta}$ has an occurrence in $A$, then it is generated by some connected component $H$ of $G_W$, unless $\min(\alpha,\beta)=1$.

\newcommand{\MaxPowers}{\mathit{MaxPowers}}
We say that $W^{\alpha,\beta}$ is a \emph{maximal} power if there is no other power $W^{\alpha',\beta'}$ in $A$ such that $\alpha' \ge \alpha$, $\beta' \ge \beta$, and $(\alpha',\beta') \ne (\alpha,\beta)$.
Similarly, we consider powers that are maximal among ones that are generated by a connected component $H$. Let $\MaxPowers_W(H)$ be the set of maximal powers generated by a connected component $H$. It can be computed in linear time using Lemma~\ref{lem:MWR} as shown in Algorithm~\ref{alg:H}, which we now explain.

For each vertex $(i,j)$ in $H$, we insert four points to a set $S$, which correspond to the four occurrences of $W$ underlying the occurrence of quartic $W^{2,2}$ at position $(i,j)$.
If $S$ is treated as a set of white cells in a grid, then $W^{\alpha,\beta}$ for $\alpha>1$ is a power generated by $H$ if and only if the grid contains a white rectangle of shape $(\alpha,\beta)$. For a cell $(i,j) \in S$, we denote $R[i,j] = \min\{p \ge 0\,:\, (i,j+p) \not\in S \}$. Assuming that the cells of $S$ are sorted by non-increasing second component, each value $R[i,j]$ can be computed from $R[i,j+1]$ in constant time, for a total of $\cO(|S|)$ time.
The sorting for all $S$ can be done globally, using radix sort.
Also, the array $R$ can be stored globally and used for all $S$, cleared after each use.
Finally, we process each maximal set of consecutive cells $(i,j),\ldots,(i+m-1,j) \in S$ that are located in the same column and apply Lemma~\ref{lem:MWR} to solve the resulting instance of the \textsc{Max White Rectangles} problem.
The total time required by this step is $\cO(|S|)$.

\begin{algorithm}[htpb]
\DontPrintSemicolon
  $S:=\emptyset$\;
  \ForEach{$(i,j)$ \KwSty{in} $V(H)$}{
    $a:=\floor{i/r}$;\ \ $b=\floor{j/c}$\;
    $S:=S\cup\{(a,b), (a+1,b), (a,b+1), (a+1,b+1)\}$\;
  }
  \vspace*{0.2cm}
  $R[0 \dd n,0 \dd n]:=(0,\ldots,0)$\;
  \ForEach{$(i,j)$ \KwSty{in} $S$ in non-increasing order of $j$}{
    $R[i,j]:=R[i,j+1]+1$\;
  }
  \vspace*{0.2cm}
  $\mathit{Result}:=\emptyset$\;
  \ForEach{maximal set $\{(i,j),(i+1,j)\ldots,(i+m-1,j)\} \subseteq S$}{
    $\Whites[1 \dd m] := R[i \dd i+m-1,j]$\;
    $\mathit{Result}:=\mathit{Result} \cup  \textsc{MaxWhiteRectangles}(\Whites)$\;\vspace*{0.3cm}
  }
  remove redundant rectangles from $\mathit{Result}$\;
  return $\mathit{Result}$\;
  \;
\caption{Computing $\MaxPowers_W(H)$ for a component $H$ of $G_W$.
}\label{alg:H}
\end{algorithm}

In the end we filter out the powers $W^{\alpha,\beta}$ that are not maximal in $A$ similarly as in the proof of Lemma~\ref{lem:MWR}, using a global array $\MaxWidth$.
Let $W^{\alpha_1,\beta_1},\ldots,W^{\alpha_k,\beta_k}$ be the resulting sequence of maximal powers, sorted by increasing first component, and let $\alpha_0=\beta_0=0$. Then the set of all quartics in $A$ with primitive root $W$ contains all $W^{2\alpha,2\beta}$ over $\alpha_{p-1} < 2\alpha \le \alpha_p,\, 1 \le 2\beta \le \beta_p$, for $p \in [2,k]$.
They can be reported in $\cO(n^2 \log^2 n)$ total time over all $W$ due to the upper bound of Theorem~\ref{thm:dist_q}.
\end{proof}

\section{Final Remarks}
We showed that the numbers of distinct runs and quartics in an $n \times n$ array are $\cO(n^2\log^2 n)$.
This improves upon previously known estimations.
We also proposed $\cO(n^2 \log^2 n)$-time algorithms for computing all occurrences of primitively rooted quartics and all distinct quartics.
A straightforward adaptation shows that for an $m \times n$ array these bounds and complexities all become $\cO(mn \log m \log n)$.

We pose two conjectures for $n\times n$ 2D-strings: 
\begin{itemize}
    \item The number of 2D-runs is $\cO(n^2)$.
    \item The number of distinct quartics is $\cO(n^2)$. 
\end{itemize}

\bibliographystyle{plainurl}
\bibliography{references}

\clearpage
\appendix

\section{Alternative Algorithm for the Proof of Lemma~\ref{lem:MWR}}\label{app:WR}
An alternative, space efficient and more direct 
algorithm that does not use additional tables $\NSVDown$ and $\NSVUp$,
is shown below. The algorithm  computes only the table $\MaxWidth$. Then, 
we can use the second phase from Algorithm~\ref{alg:NSV2}.
We assume that the table $\MaxWidth$ is initially filled with zeros.

\begin{algorithm}[htpb]
\DontPrintSemicolon
  $\Whites[0] := \Whites[m+1] :=
  0\;$\;
  $S:=$ empty stack; $\ppush(S,0)$\;\vspace*{0.1cm}
  \For{$i:=m$ \KwSty{down to} $0$}{\vspace*{0.1cm}
    \While{$\Whites[i] < \Whites[\ttop(S)]$}{
      $k:= \ttop(S)$;\
      $h := \ttop(S)-i-1$\;
      $\MaxWidth[h]:=\max(\MaxWidth[h],\Whites[k])$\;
       $\ppop(S)$\;
\vspace*{0.1cm}    }
    \lIf{$\Whites[\ttop(S)]=\Whites[i]$}{$\ppop(S)$}
    $\ppush(S,i)$\;
  }
  \;
\caption{Alternative implementation of the first phase
in Algorithm~\ref{alg:NSV2}.}
\label{alg:NSV}
\end{algorithm}

The algorithm is a
version of a folklore  
 algorithm for the Nearest Smaller  Values problem and correctness can be shown
 using the same arguments.
 If $\Whites[i]<\Whites[i+1]$, then the algorithm  produces
 shapes of 
 all Max White Rectangles
anchored at $i+1$,
otherwise $i+1$ is ``nonproductive''. 
Observe that  $i+1=\ttop(S)$ when we start processing $i\ge 1$.

Let us analyze the time complexity of the algorithm. In total $m+2$ elements are pushed to the stack. Each iteration of the while-loop pops an element, so the total number of iterations of this loop is $\cO(m)$. Consequently, the algorithm works in $\cO(m)$ time. In the end one needs to filter out non-maximal rectangles as in the 
previous proof of Lemma~\ref{lem:MWR}.

\end{document}